\documentclass[11pt]{article}
\usepackage[margin=1.1in]{geometry}
\usepackage{amsmath,amssymb,amsthm}
\usepackage[numbers,sort&compress]{natbib}
\usepackage{graphicx}
\usepackage{float}
\usepackage{tikz}
\usetikzlibrary{arrows.meta,positioning,fit,shapes.geometric}
\usepackage[colorlinks=true,linkcolor=blue,citecolor=blue,urlcolor=blue]{hyperref}

\newtheorem{theorem}{Theorem}
\newtheorem{proposition}{Proposition}
\newtheorem{corollary}{Corollary}

\theoremstyle{definition}
\newtheorem{definition}{Definition}
\newtheorem{example}{Example}
\theoremstyle{remark}
\newtheorem{remark}{Remark}

\newcommand{\dist}{\mathrm{dist}}
\newcommand{\cone}{\mathrm{cone}}
\newcommand{\crew}{\mathrm{crew}}
\newcommand{\STOP}{\mathtt{STOP}}
\newcommand{\R}{\mathbb{R}}

\title{Collective Counterfactual Planning:\\ Coordination, Consent, and Verification\\ under Representational Constraints}
\author{Chainarong Amornbunchornvej\\
\small National Electronics and Computer Technology Center (NECTEC), NSTDA\\
\small Pathum Thani, Thailand \qquad \texttt{[chainarong.amo@nectec.or.th]}}

\begin{document}
\maketitle

\begin{abstract}
Groups routinely complete projects that no single member can plan, execute, or verify alone. We propose a formal model of this phenomenon, \emph{Collective Counterfactual Planning} (CCP), in which the binding limitation on each agent is neither capability, nor knowledge, nor observability, but \emph{representational geometry}: each agent $i$ perceives the shared state, conceives prospective moves, consents to actions, and certifies goal requirements only through an orthogonal projection onto an agent-specific subspace $V_i$ of a common task space (the subspaces themselves are common knowledge; only consent thresholds are private). Four gates --- the exogenous implementation coalitions $R(u)$, together with three representational gates: conception $C_i(x,O)$, consent thresholds on interpretability, and task-relative verification qualification $Q_k$ --- jointly determine whether a team can reach a conjunctive goal $O=\bigcap_k O_k$ \emph{and legitimately recognize that it has done so}. We define the Collective Counterfactual Solvability (CCS) problem and separate three semantic levels: geometric feasibility, executable attainment, and validated completion. The main results expose a positive--negative duality. Iterated cross-agent relay can unlock a solution that no one-shot pooling of individual plans contains, but any goal requirement depending essentially on the subspace dark to the entire team is unverifiable and therefore not validly completable, even when the physical trajectory accidentally attains it. Memoryless and audited consent further constrain different objects --- action directions versus cumulative trajectory states --- and neither dominates the other. On the constructive side, a four-step exhaustive horizon-bounded solvability scheme (coverage preflight, relay closure, ratification, terminal inspection) is sound and complete under exact representation of the complete relay closure; restricted implementations remain sound on returned plans but need not be complete. The model thereby gives one geometry for sequential mutual enabling, competent execution of steps whose purpose is invisible to the executor, forced sub-teaming at expertise boundaries, and completion that cannot be validly declared.
\end{abstract}

\section{Introduction}\label{sec:intro}

Consider building a house. No individual on the project --- architect, structural engineer, electrician, plumber, inspector --- represents the entire task. Each understands some aspects of the current state of the project and is blind to others; at each stage, someone sees a useful next move that others cannot yet see; every move requires the participation of specific tradespeople, each of whom will act only on work they sufficiently understand; and the project is finished not when the building happens to satisfy the code, but when each requirement has been \emph{certified by someone qualified to check it}. The team does not need to be able to build every possible house. It needs a collection of people whose combined representational and practical competencies suffice to plan, execute, and verify \emph{this} house.

This paper proposes a minimal formal model of that situation. Its central commitment is that the limitation doing the work is \emph{representational}: what an agent can perceive, think of, agree to, and certify is bounded by the subspace of the task's state space that the agent can represent. This is a deliberately different modeling choice from the three dominant traditions in multi-agent planning, which locate agent limitation in \emph{capability} --- which operators an agent owns \citep{brafman2008one,torreno2017cooperative} --- in \emph{knowledge} --- which epistemic states an agent can distinguish \citep{bolander2011epistemic,engesser2017cooperative} --- or in \emph{observability} --- which signals an agent receives \citep{bernstein2002complexity,oliehoek2016concise}. The point of the contrast is not that those frameworks cannot model persistent failure; each is expressive. The point is that CCP holds those channels conceptually separate and makes representational structure itself the explicit object of analysis --- and, so isolated, it yields limits of its own kind: failures that persist because the missing dimension is not information an agent lacks but structure the agent cannot host.

The linear-algebraic substrate follows the cognitive-geometric framework of \cite{amornbunchornvej2026cognitive} and is consonant with the broader case that concepts and cognitive states are usefully modeled as vectors in structured spaces \citep{gardenfors2004conceptual,kriegeskorte2013representational,piantadosi2024concepts,pmlr-v235-park24c}. Each agent $i$ carries a subspace $V_i$ of a common task space. Four gates then govern collective action, and they are not of one kind. \emph{Implementation} is exogenous task structure: each action names the coalition of agents physically or institutionally required to perform it, and nothing about $R(u)$ is derived from geometry. The other three are induced by the agents' representational geometry: \emph{conception} (an agent can generate a move only if it makes apparent progress in the agent's projection of the task), \emph{consent} (a required agent participates only in moves that are sufficiently interpretable to that agent), and \emph{verification} (an agent is qualified to certify a goal requirement only if all distinctions relevant to that requirement are visible in the agent's subspace).

\paragraph{Contributions.}
\begin{enumerate}
\item \textbf{Model.} A parsimonious formal setting (Section~\ref{sec:model}) that holds the exogenous implementation structure $R(u)$ separate from three gates induced by the agents' representational geometry --- conception, consent, and qualification; none of the four implies another.
\item \textbf{Problem.} The Collective Counterfactual Solvability problem (Section~\ref{sec:problem}), with three semantic levels separating \emph{geometric} feasibility, executable \emph{attainment}, and \emph{validated} completion. Objective attainment and validated completion come apart by design.
\item \textbf{Collective reach and collective blindness.} Iterated cross-agent re-expansion can make a solution available that is absent from every one-shot individual tree, and the enabling signal must pass through a component visible to the next agent (Theorem~\ref{thm:relay}, Corollary~\ref{cor:visible}). Dually, a requirement that depends essentially on the team's dark subspace has no qualified verifier and cannot be validly completed (Theorem~\ref{thm:dark}).
\item \textbf{Consent geometry.} Memoryless consent gates action directions, whereas audited consent gates cumulative trajectory states. The resulting solvability notions are incomparable (Theorems~\ref{thm:dichotomy}--\ref{thm:separation}), and the same geometry can force non-plenary sub-teaming (Proposition~\ref{prop:cones}).
\item \textbf{Exhaustive solvability scheme and correctness.} A four-step exhaustive horizon-bounded scheme --- coverage preflight, relay closure, ratification, terminal inspection --- is sound and complete under exact representation of the complete relay closure (Theorem~\ref{thm:correctness}). A restricted practical search is still sound on any returned plan, but may miss existing solutions (Corollary~\ref{cor:restricted}).
\end{enumerate}

The model is deliberately non-strategic: agents are truthful, share the goal, and consent mechanically. What remains when incentives, deception, and disagreement are all removed is the paper's subject --- the purely representational obstacles to collective planning, and the sense in which teamwork, sub-teaming, and institutionalized verification are forced by geometry rather than chosen by convention.

\section{Related work}\label{sec:related}

\paragraph{Multi-agent planning.}
MA-STRIPS and its successors model heterogeneity through operator ownership \citep{brafman2008one}; cooperative multi-agent planning is surveyed in \citet{torreno2017cooperative}; the complexity landscape of classical planning is anchored by \citet{bylander1994computational} and its textbook treatment \citep{ghallab2016automated}. Epistemic planning locates the limitation in knowledge, with Dynamic Epistemic Logic as the standard vehicle \citep{bolander2011epistemic}. Closest to our relay phenomenon is \emph{implicit coordination} \citep{engesser2017cooperative}, where perspective shifts allow one agent to plan on the assumption that another will recognize and continue the plan; the mechanism there is epistemic indistinguishability, whereas ours is projection geometry, and no analogue of our span-based impossibility or verification qualification arises. Decentralized POMDPs locate the limitation in observability and yield the field's canonical hardness results \citep{bernstein2002complexity,oliehoek2016concise}. Required-cooperation analysis \citep{zhang2016required} characterizes when cooperation is unavoidable under \emph{capability} heterogeneity; our Proposition~\ref{prop:cones} is the representational counterpart. Team formation with required skills \citep{lappas2009finding,juarez2021comprehensive} is the discrete ancestor of the team-synthesis problem our model raises but defers.

\paragraph{Coordination and collective agency.}
The classical accounts locate what coordination needs in salience
\citep{schelling1980strategy}, equilibrium structure
\citep{Cooper1999,VanHuyck1990}, or common knowledge and its fragility
\citep{rubinstein1989electronic,fagin1995reasoning}; group agency is treated in
\citet{ListPettit2011}, cooperative problem solving in
\citet{WooldridgeJennings1999,Wooldridge2009}, and open problems in cooperative
AI in \citet{Dafoe2020}. Agreement dynamics on time-varying interaction
networks provides a complementary mathematical account of collective
convergence \citep{chazelle2011total}. Amornbunchornvej et al.\ formalize
coordination events, initiators, and following mechanisms from dynamic
interaction data \citep{10.1145/3201406}; Amornbunchornvej and Berger-Wolf
model heterogeneous following strategies through which individuals respond
to neighbors, particular individuals, or combinations of influences in
producing group-level coordination
\citep{amornbunchornvej2020framework}. Consensus among agents with dissimilar
cognitive architectures is studied information-theoretically by
\citet{SowinskiEtAl2022}, and mutual-intelligibility constraints between
linguistic agents by \citet{KomarovaNiyogi2004}. Each of these presupposes the
representational frame within which focal points are focal and messages are
messages; the present model makes that presupposed layer the variable.
Mathematically, the nearest neighbor is sheaf-theoretic coordination over
heterogeneous stalks with linear restriction maps
\citep{hansen2021opinion}; the aims differ --- consensus dynamics there,
planning, executability, and verification here.

\paragraph{The phenomenon.}
That groups accomplish what no member represents is documented in distributed cognition \citep{hutchins1995cognition} and epistemic dependence \citep{hardwig1985epistemic,hardwig1991role}; ``who knows what'' as a team resource is the subject of transactive memory \citep{Lewis2003TransactiveMemory}. On the organizational side, bounded rationality \citep{rogow1957models}, organization design as information processing under cognitive limits \citep{RePEc:inm:orinte:v:4:y:1974:i:3:p:28-36}, and organizations as interpretation systems \citep{daft1984toward} anticipate the sub-teaming result; verification as institution \citep{Power1997AuditSociety} anticipates the coverage preflight; and invisible or normalized failure \citep{perrow1984normal,vaughan1996challenger,WeickSutcliffeObstfeld1999CollectiveMindfulness} is the empirical face of representationally undetectable failure. Shared intentionality and common ground \citep{tomasello2005shared,clark1996using} ground the shared task and the possibility of joint action.

\paragraph{Terminology.}
``Counterfactual'' here means \emph{within-basis prospective rollout} --- forward simulation of moves from the current state --- and is distinct from the backward-looking causal senses of \citet{pearl2009causality} and \citet{lewis1973counterfactuals}, as well as from models in which the representational basis itself changes. The psychology of such forward simulation is prospection \citep{gilbert2007prospection}. Actions as vectors in conceptual spaces are treated by \citet{gardenfors2012using}, for single events rather than multi-agent planning.

\section{The model}\label{sec:model}

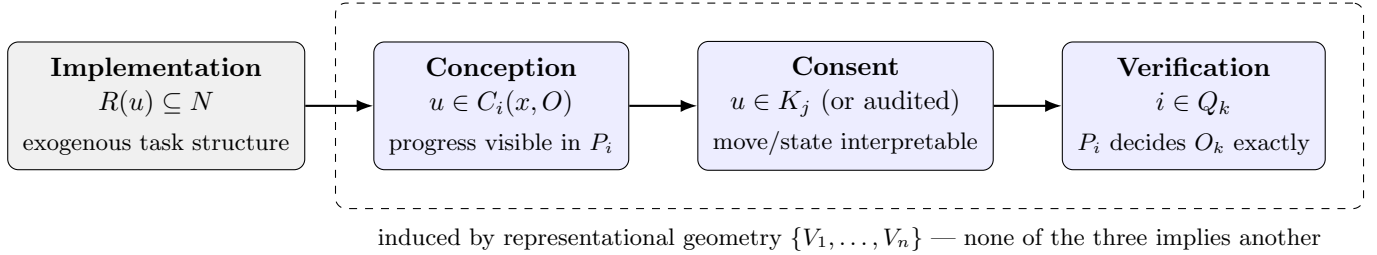
\begin{figure}[H]
\centering
\begin{tikzpicture}[
  node distance=6mm and 9mm,
  font=\small,
  box/.style={draw, rounded corners, align=center, minimum width=2.55cm, minimum height=1.5cm, inner sep=2mm},
  arr/.style={-{Latex[length=2.2mm]}, thick}
]
\node[box, fill=gray!12] (R) {\textbf{Implementation}\\[1pt] $R(u)\subseteq N$\\[3pt] \footnotesize exogenous task structure};
\node[box, fill=blue!7, right=of R] (C) {\textbf{Conception}\\[1pt] $u\in C_i(x,O)$\\[3pt] \footnotesize progress visible in $P_i$};
\node[box, fill=blue!7, right=of C] (K) {\textbf{Consent}\\[1pt] $u\in K_j$ (or audited)\\[3pt] \footnotesize move/state interpretable};
\node[box, fill=blue!7, right=of K] (Q) {\textbf{Verification}\\[1pt] $i\in Q_k$\\[3pt] \footnotesize $P_i$ decides $O_k$ exactly};
\draw[arr] (R) -- (C);
\draw[arr] (C) -- (K);
\draw[arr] (K) -- (Q);
\node[draw, dashed, rounded corners, fit=(C)(K)(Q), inner sep=5mm] (geomfit) {};
\node[below=1mm of geomfit.south, font=\footnotesize, align=center] {induced by representational geometry $\{V_1,\dots,V_n\}$ --- none of the three implies another};
\end{tikzpicture}
\caption{The four gates of CCP. Implementation $R(u)$ is exogenous: it fixes which coalition may perform a direction, independent of anyone's geometry. The remaining three gates are all induced by the agents' subspaces but ask different questions: conception asks whether a move is visible progress \emph{to the agent who would propose it}; consent asks whether the required implementers find the move (memoryless) or the resulting cumulative state (audited) sufficiently interpretable; verification asks whether some agent's projection decides a requirement exactly. A move must clear implementation, conception, and consent to execute (Section~\ref{sec:model}); a goal must clear verification, at every requirement, to be validly declared complete (Definition~\ref{def:stop}).}
\label{fig:gates}
\end{figure}

\subsection{Agents, task, actions}

\begin{definition}[Agents and spans]\label{def:agents}
$N=\{1,\dots,n\}$ agents act in an ambient task space $V=\R^d$, the frame induced by the goal. Each agent $i$ carries a subspace $V_i\subseteq V$ --- the directions $i$ can represent --- with orthogonal projection $P_i$ onto $V_i$ and complement projection $P_i^{\perp}$. The threshold $\tau_i\in(0,1]$ is the minimum interpretability $i$ requires before consenting; its angular form is $\theta_i=\arccos\sqrt{\tau_i}$, with bundling constant $\kappa(\tau_i)=\tan\theta_i$. The group span is $S=V_1+\cdots+V_n$ --- everything somebody can see; its orthogonal complement $S^{\perp}$ is the \emph{dark subspace} --- what nobody can see. Point-to-set distance is $\dist(y,A)=\inf_{a\in A}\|y-a\|$.
\end{definition}

\begin{definition}[Task]\label{def:task}
The start is $x_0\in V$; the goal is a public conjunction of requirements $O=\bigcap_{k=1}^{m}O_k$ with $O_k\subseteq V$; the demand set is $\Delta=O-x_0$. Different requirements may concern different subspaces; all competence notions below are relative to this task only.
\end{definition}

\begin{definition}[Actions as directions]\label{def:actions}
The resource $U$ is a finite set of unit vectors in $V$. Each direction $u\in U$ carries a required coalition $R(u)\subseteq N$, its \emph{necessary implementers}: $u$ can be performed only by these agents acting jointly; if any one is missing, $u$ is unavailable. A move is a token $(u,\lambda)$ with duration $\lambda>0$, under jump dynamics $x_{t+1}=x_t+\lambda_t u_t$.
\end{definition}

\subsection{Conception}

\begin{definition}[Counterfactual availability]\label{def:conception}
Agent $i$ can conceive $u$ at state $x$ iff moving along $u$ makes progress toward the task \emph{in $i$'s projection}:
\begin{equation}\label{eq:C}
u\in C_i(x,O)\iff \exists\lambda>0:\ \dist\big(P_i(x+\lambda u),\,P_iO\big)<\dist\big(P_i x,\,P_iO\big).
\end{equation}
\end{definition}

Progress here is apparent, not actual: projected progress need not be progress in $V$. Two consequences are free from the definition: $C_i(x,O)$ depends on $x$ only through $P_i x$ --- conception is computed inside the agent's subspace --- and $P_i u=0\Rightarrow u\notin C_i(x,O)$: no one conceives a move he cannot see at all. Conception and consent are distinct gates: $u\in C_i(x,O)$ means $i$ can \emph{see the point} of the move from here; $u\in K_i$ (below) means $i$ \emph{interprets} it sufficiently to consent to executing it. Neither implies the other.

\subsection{Interpretability and consent}

\begin{definition}[Interpretability; consent cones]\label{def:interp}
The interpretability of a vector $v\neq0$ to agent $j$ is
\begin{equation}\label{eq:I}
I_j(v)=\frac{\|P_j v\|^2}{\|v\|^2}=\cos^2\angle(v,V_j),
\end{equation}
the fraction of $v$'s squared length lying inside $j$'s subspace; $I_j(\lambda v)=I_j(v)$ for $\lambda>0$ --- magnitude is invisible to interpretability. By convention, $I_j(0)=1$. The consent cone is
\[
K_j=\big\{\,v\in V:\ \|P_j^{\perp}v\|\le\tan\theta_j\,\|P_jv\|\,\big\};
\]
for $v\neq0$, $v\in K_j\iff I_j(v)\ge\tau_j\iff\angle(v,V_j)\le\theta_j$. $K_j$ is a nonconvex double cone: $I_j(v)=I_j(-v)$, so consent sees axes, not arrows.
\end{definition}

\begin{definition}[Consent models; executability]\label{def:consent}
Under \textbf{memoryless} consent, $j$ consents to $(u,\lambda)$ iff $u\in K_j$; duration plays no role --- consenting to a direction is consenting to any distance along it. Under \textbf{audited} consent, at cumulative displacement $D$, $j$ consents iff the new total stays interpretable: $I_j(D+\lambda u)\ge\tau_j$, audited at jump endpoints. Implementation requires the participation of every necessary agent, and each participates only in steps he interprets: $(u,\lambda)$ \emph{executes} at $D$ iff every $j\in R(u)$ consents under the chosen model. Sovereign execution: consent has no substitute.
\end{definition}

\subsection{Plans, perception, knowledge}

\begin{definition}[Plans]\label{def:plans}
A plan is $p=((u_1,\lambda_1),\dots,(u_T,\lambda_T))$ with skeleton $w=(u_1,\dots,u_T)$; coalitions $R(u_t)$ are public annotation. The crew is $\crew(p)=\bigcup_{t\le T}R(u_t)$, the personnel roster. Consent is strictly per-step: step $t$ needs yeses from $R(u_t)$ only; nobody ever consents to a plan. The trajectory is $x_t=x_0+D_t$ with $D_t=\sum_{s\le t}\lambda_s u_s$ and final state $x_T$. Durations are chosen at extension time (availability is evaluated at realized states), so plans are generated already instantiated. Tokens are public: every agent computes $P_jD_t$ himself; the executed past is public arithmetic.
\end{definition}

\paragraph{Perception and knowledge.}
Agent $i$ perceives the state and each requirement only through his projection, $P_i(x_0+D)$ and $P_iO_k$; components in $V_i^{\perp}$ are dark to $i$. All subspaces $V_1,\dots,V_n$ are public (common knowledge); the thresholds $\tau_j$ are the model's only private information. All agents are truthful: queried consent values and reported inspections are returned as they are. The model concerns representational limits, not strategic behavior.

\subsection{Verification}

\begin{definition}[Task-relative qualification]\label{def:Q}
The qualified verifiers of requirement $O_k$ are
\begin{equation}\label{eq:Q}
Q_k=\big\{\,i\in N:\ \forall x,y\in V,\ \ P_i x=P_i y\ \Rightarrow\ (x\in O_k\iff y\in O_k)\,\big\}.
\end{equation}
Equivalently, $i\in Q_k$ iff $O_k=P_i^{-1}(P_iO_k)$: membership in $O_k$ is completely determined by $i$'s projection, and the qualified check $P_i x\in P_iO_k$ is \emph{sound and complete for that requirement}.
\end{definition}

Qualification is task-relative only: $i\in Q_k$ implies neither $V_i=V$ nor qualification for any other requirement; different requirements may have different qualified verifiers. Competence is not completion: $Q_k\neq\varnothing$ means someone \emph{can decide} $O_k$, not that $O_k$ holds.

\begin{definition}[Cover; executable stopping]\label{def:stop}
The verification cover holds iff $Q_k\neq\varnothing$ for every $k$ --- a state-independent property of the team--task pair. Executable stopping is
\begin{equation}\label{eq:stop}
\STOP(x)\iff \forall k\ \exists i\in Q_k:\ P_i x\in P_iO_k .
\end{equation}
Under the cover, $\STOP(x)\iff x\in O$: qualified checks decide the objective exactly. Without the cover, landing in $O$ is invisible as an achievement. \textbf{Objective attainment $\neq$ validated project completion.}
\end{definition}

\subsection{Plan language: relay closure}

\begin{definition}[Extension; closure]\label{def:closure}
Agent $i$ may append token $(u,\lambda)$ to a plan whose terminal state is $x$ iff: (1) $i\in R(u)$ --- the extender is one of $u$'s necessary implementers; (2) $u\in C_i(x,O)$, with $\lambda$ chosen so that $\dist(P_i(x+\lambda u),P_iO)<\dist(P_ix,P_iO)$ --- the move is sized to make apparent progress; (3) under memoryless consent, $u\in K_i$. Extenders leave no trace: who thought of a step is not part of $p$. The closure: each agent grows his counterfactual tree from $x_0$ and publishes it whole; take the union; every agent re-expands from every new node's state; repeat to a fixed point, to horizon $h$ (a depth bound; results are $h$-relative); write $L_h$ for this horizon-$h$ closure.
\end{definition}

Prefix-dependence is real: availability at a node depends on the state its prefix created, so $L_h$ has no closed form over a fixed alphabet --- the relay iteration is essential, not distributed bookkeeping. In the comparison variant with exogenously state-independent availability, order-irrelevance returns for the memoryless model and $T\le d$ suffices (Proposition~\ref{prop:special}).

\section{The problem}\label{sec:problem}

\begin{definition}[CCS]\label{def:ccs}
\textsc{Collective Counterfactual Solvability}.\\
Given $(V,\{V_i\},\{\tau_i\},U,R,x_0,\{O_k\}_{k=1}^m;$ consent model; horizon $h)$: does a plan $p\in L_h$ exist such that every step executes and $\STOP(x_T)$ holds? The synthesis version returns $p$ or $\bot$.
\end{definition}

\paragraph{Three semantic levels.}
(1) \emph{Geometric}: $\Delta\cap\cone(U)\neq\varnothing$, where $\cone(U)$ is the set of nonnegative combinations of directions --- could the available physical moves ever add up? (2) \emph{Attainment}: an executable relay plan lands in $O$ --- the group can conceive and execute a route to the objective. (3) \emph{Validated} ($=$ CCS yes): attainment together with a verification cover --- the group can also legitimately recognize that it is done. The gap between levels 1 and 2 is created by conception and consent; the gap between 2 and 3 is created by verification. \emph{Discoverability} is treated separately as an algorithmic property: for an exhaustive horizon-bounded search it coincides with validated solvability by Theorem~\ref{thm:correctness}, while a restricted practical search may be incomplete.

\section{Results}\label{sec:results}

All results are stated for the model of Section~\ref{sec:model}. We first fix a running example; all its numerical claims are machine-checked.

\begin{example}[E1]\label{ex:e1}
Two agents, two actions, and a two-part goal in the plane ($d=2$).
\begin{enumerate}
\item \textbf{Agents.} Agent 1 sees only the horizontal axis ($V_1=\mathrm{span}(e_1)$); agent 2 sees only the vertical axis ($V_2=\mathrm{span}(e_2)$). Both demand $\tau_1=\tau_2=0.8$, a consent half-angle of $\theta\approx26.57^{\circ}$.

\item \textbf{Actions.} $U=\{u_1,u_2\}$. Direction $u_1=(\cos\alpha,\sin\alpha)$ with $\alpha=20^{\circ}$ is nearly horizontal but slightly tilted, and only agent 1 can perform it ($R(u_1)=\{1\}$). Direction $u_2=(0,-1)$ points straight down, and only agent 2 can perform it ($R(u_2)=\{2\}$).

\item \textbf{Task.} $O_1=\{x:\langle x,e_1\rangle=1\}$: bring the horizontal coordinate to $1$. $O_2=\{x:\langle x,e_2\rangle=0\}$: end with the vertical coordinate at $0$. So $O=\{(1,0)\}$ and $x_0=0$. Each requirement is cylindrical over its owner's axis, so $1\in Q_1$ and $2\in Q_2$: the cover holds, and each agent certifies exactly his own requirement.

\item \textbf{Why agent 2 cannot begin.} At $x_0$ the vertical coordinate is already $0$: agent 2's projected distance to his goal is zero, no move can strictly improve it, and so $u_2\notin C_2(x_0,O)$. From where agent 2 stands, the job looks finished before it has started.

\item \textbf{Agent 1's step.} Agent 1 conceives $u_1$ (it reduces his horizontal distance, so $u_1\in C_1(x_0,O)$) and consents to it ($I_1(u_1)=\cos^2 20^{\circ}\approx0.883\ge0.8$). Choosing $\lambda_1=1/\cos\alpha$ lands his coordinate exactly on target: the state becomes $x_1=(1,\tan\alpha)\approx(1,0.364)$. The tilt is the point: since $P_2u_1=\sin\alpha\neq0$, the step disturbs the vertical coordinate that agent 2 watches.

\item \textbf{The flip.} At $x_1$ agent 2 now sees a problem ($0.364\neq0$), so $u_2\in C_2(x_1,O)$: a move inconceivable at $x_0$ has become conceivable because of someone else's action. Agent 2 consents ($I_2(u_2)=1$) and executes $\lambda_2=\tan\alpha$, reaching $x_T=(1,0)\in O$.

\item \textbf{Finish.} Agent 1 inspects $O_1$, agent 2 inspects $O_2$; both pass, so $\STOP(x_T)$ holds. All quantities are machine-checked. This one instance witnesses the relay of Theorem~\ref{thm:relay}, the loop inversion of Proposition~\ref{prop:loop}, and the memoryless direction of Theorem~\ref{thm:separation}.
\end{enumerate}
\end{example}

\begin{theorem}[Unverifiability from darkness]\label{thm:dark}
Suppose some requirement $O_k$ depends essentially on $S^{\perp}$: there exist $x\in O_k$ and $x'=x+z\notin O_k$ with $z\in S^{\perp}\setminus\{0\}$. Then $Q_k=\varnothing$; hence no verification cover exists and CCS $=$ no --- for every $x_0$, $U$, $R$, and consent model, even if the physical trajectory lands in $O$.
\end{theorem}
\begin{proof}
For every $i$, $V_i\subseteq S$, so $S^{\perp}\subseteq V_i^{\perp}$ and $P_iz=0$; hence $P_ix=P_ix'$ while $x\in O_k$ and $x'\notin O_k$. No $i$ satisfies \eqref{eq:Q}, so $Q_k=\varnothing$. $\STOP$ requires a qualified certifier for every requirement, so it is false at every state, and no plan satisfies Definition~\ref{def:ccs}.
\end{proof}

\begin{remark}[Conception is equally blind]\label{rem:blindconception}
If $u\in U$ lies wholly in $S^{\perp}$ then $P_iu=0$ for all $i$, so by Definition~\ref{def:conception} $u$ is never conceivable and never appendable: deliberate motion along $S^{\perp}$ occurs only as a side-component of visible directions. Unsolvability and its undetectability coincide on $S^{\perp}$.
\end{remark}

\begin{theorem}[Relay]\label{thm:relay}
There exist instances on which the one-shot union of the agents' counterfactual trees at $x_0$ contains no plan reaching $O$, while the fixed-point closure contains a plan achieving $\STOP$.
\end{theorem}
\begin{proof}
Example E1. \emph{One-shot}: agent 2's tree is the bare root --- the only direction with $2\in R$ is $u_2$, and $u_2\notin C_2(x_0,O)$; agent 1's tree contains only $u_1$-chains, and after any prefix the $e_2$-coordinate equals $(\sum\lambda)\sin\alpha>0$, so no node lies in $O$. \emph{Closure}: agent 1 publishes the branch to $x_1=(1,\tan\alpha)$; re-expanding from $x_1$, agent 2 has $u_2\in C_2(x_1,O)$ with $\lambda_2=\tan\alpha$ realizing projected progress $\tan\alpha\to0$, $2\in R(u_2)$, and $u_2\in K_2$; the extended plan reaches $(1,0)\in O$; every memoryless consent passes; the cover holds and both inspections pass.
\end{proof}

\begin{corollary}[Relays propagate only through visible components]\label{cor:visible}
Since $C_i$ depends on the state only through $P_ix$, $P_2u_1=0\Rightarrow P_2x_1=P_2x_0\Rightarrow C_2(x_1,O)=C_2(x_0,O)$: a move completely invisible to agent 2 cannot make any new counterfactual move available to agent 2. Genuine relay requires the preceding action to alter some component visible in the next agent's representational space. In E1 the flip exists exactly because $P_2u_1=\sin\alpha\neq0$.
\end{corollary}

\begin{theorem}[Dichotomy; visitation-local audit]\label{thm:dichotomy}
(i) Memoryless consent imposes no magnitude constraint on an accepted direction: for any $j\in R(u)$ with $u\in K_j$, consent to $(u,\lambda)$ holds for every $\lambda>0$, by scale invariance $I_j(\lambda u)=I_j(u)$. Consent alone therefore never bounds $\|P_j^{\perp}D\|$ whenever some accepted direction has $P_j^{\perp}u\neq0$: auditing is consent's only brake on magnitude. (The conception rule's progress clause does bound the \emph{chosen} duration at generation --- a conception constraint, not a consent one; the claim here concerns consent only.)
(ii) Under audited consent, for every step $t$ with $j\in R(u_t)$, consent is equivalent to
$\|P_j^{\perp}D_t\|\le\tan\theta_j\,\|P_jD_t\|$.
The constraint is visitation-local: it binds only at $j$'s required steps; it is enforced at all $t$ iff $j$ is a plenary auditor ($j\in R(u_t)$ for every $t$); and the endpoint is constrained by $j$ iff $j\in R(u_T)$.
\end{theorem}
\begin{proof}
(i) is the scale-invariance of $I_j$ in \eqref{eq:I}, read at the consent gate. (ii): for $D\neq0$, $I_j(D)\ge\tau\iff\|P_jD\|^2\ge\tau(\|P_jD\|^2+\|P_j^{\perp}D\|^2)\iff\|P_j^{\perp}D\|^2\le\frac{1-\tau}{\tau}\|P_jD\|^2$, and $\sqrt{(1-\tau)/\tau}=\tan\theta_j$; at $D=0$ both sides hold trivially, using the convention $I_j(0)=1$. The locality clauses restate Definition~\ref{def:consent}.
\end{proof}

\begin{example}[E2]\label{ex:e2}
$d=2$, $V_1=\mathrm{span}(e_1)$, $V_2=\mathrm{span}(e_2)$, $\tau_1=0.8$, $\tau_2=0.05$; $U=\{e_1,e_2\}$ with $R(e_1)=\{1\}$, $R(e_2)=\{1,2\}$; $O_1=\{x:\langle x,e_1\rangle=1\}$, $O_2=\{x:\langle x,e_2\rangle=0.3\}$; $x_0=0$. The cover holds ($1\in Q_1$, $2\in Q_2$).
\end{example}

\begin{theorem}[Consent-model incomparability]\label{thm:separation}
Neither consent model's solvable set contains the other: E1 is memoryless-solvable and audited-unsolvable, while E2 is audited-solvable and memoryless-unsolvable.
\end{theorem}
\begin{proof}
\emph{E2, audited-solvable:} the plan $(e_1,1),(e_2,0.3)$ is generated (agent 1 extends the first step with $P_1$-progress to $0$; agent 2 extends the second, $e_2\in C_2$ at $(1,0)$ with $\lambda=0.3$ realizing progress) and ratified: $I_1((1,0))=1$, $I_1((1,0.3))=1/1.09\approx0.917\ge0.8$, and $I_2((1,0.3))\approx0.083\ge0.05$; the endpoint lies in $O$ and both inspections pass. \emph{E2, memoryless-unsolvable:} any plan reaching $O$ must change the $e_2$-coordinate, hence contains an $e_2$ step, whose coalition includes agent 1; but $I_1(e_2)=0<\tau_1$, so agent 1 never consents memorylessly.

\emph{E1, memoryless-solvable} is Theorem~\ref{thm:relay}. \emph{E1, audited-unsolvable:} note first that every plan in $L_h$ begins with $u_1$: at $x_0$, $C_2(x_0,O)=\varnothing$ blocks $u_2$, and $2\notin R(u_1)$ blocks agent 2 entirely; consequently the cumulative $e_1$-coordinate is at least $\lambda_{\mathrm{first}}\cos\alpha>0$ forever, since $u_2$ has zero $e_1$-component and $u_1$ adds positively. Reaching $O$ requires final $e_2$-coordinate $0$, so a plan reaching $O$ contains a last $u_2$ step; let $v$ be the $e_2$-coordinate immediately after it and $\mu\ge0$ the total $u_1$-duration after it. Then $v=-\mu\sin\alpha$, and the $e_1$-coordinate at that step is $1-\mu\cos\alpha>0$ by the first-step argument.
If $\mu=0$, the $u_2$ step is last and its audit is $I_2(D_T)=I_2((1,0))=0<\tau_2$; reject.
If $\mu>0$, the audit of that $u_2$ step requires $|v|\ge\cot\theta_2\cdot(1-\mu\cos\alpha)=2(1-\mu\cos\alpha)$. But rule (2) of Definition~\ref{def:closure} --- the chosen duration realizes strict projected improvement --- gives $|v|<|e_2\text{ before the step}|$, and the $e_2$-coordinate before any $u_2$ step is at most $A\sin\alpha$ with $A=(1-\mu\cos\alpha)/\cos\alpha$ the prior $u_1$-duration (earlier $u_2$ steps only shrink $|e_2|$, again by rule (2)). Hence $|v|<(1-\mu\cos\alpha)\tan\alpha$, and the audit demands $(1-\mu\cos\alpha)\tan\alpha>2(1-\mu\cos\alpha)$, i.e.\ $\tan\alpha>2$ --- false at $\alpha=20^{\circ}$.
\end{proof}

\begin{remark}[Mechanism]\label{rem:mechanism}
The two models control different objects. Memoryless consent gates \emph{directions}: an uninterpretable direction is vetoed in every context, however small the step (E2). Audited consent gates \emph{cumulative trajectories}: it forbids interpretable directions from accumulating into uninterpretable positions (E1), yet permits an uninterpretable direction to ride inside interpretable cumulative mass --- in E2, $I_1(e_2)=0$ while $I_1((1,0.3))\approx0.917$. Neither model is simply stricter; they are two distinct readings of ``I only join what I understand.''
\end{remark}

\begin{remark}\label{rem:clause}
The progress-realizing clause of Definition~\ref{def:closure} is essential: without it, an overshooting $u_2$ duration (large $|v|$) satisfies the audit and the separation fails. The separation holds for any witness angle with $\tan\alpha\le\cot\theta_2$ ($=2$ at $\tau_2=0.8$).
\end{remark}

\begin{remark}[Adaptive consent]\label{rem:adaptive}
Definitions~\ref{def:consent} and~\ref{thm:separation} fix a single consent model for the whole plan. Nothing in the model forces this: a plan may instead carry a \emph{consent-policy sequence} $\sigma=(\sigma_1,\dots,\sigma_T)\in\{\mathrm M,\mathrm A\}^T$, with step $t$ ratified by
\[
\mathrm{Consent}_j(t)=
\begin{cases}
u_t\in K_j, & \sigma_t=\mathrm M,\\
I_j(D_t)\ge\tau_j, & \sigma_t=\mathrm A,
\end{cases}
\qquad j\in R(u_t).
\]
The two pure regimes are the constant sequences $\sigma\equiv\mathrm M$ and $\sigma\equiv\mathrm A$; nothing else in Steps~0--1 or~3 changes, and generation (Definition~\ref{def:closure}) is untouched --- $\sigma$ governs ratification only. Write $\mathcal S_{\mathrm M},\mathcal S_{\mathrm A},\mathcal S_{\mathrm{adaptive}}\subseteq L_h$ for the sets of plans ratifiable under, respectively, the pure memoryless model, the pure audited model, and \emph{some} policy $\sigma$. Trivially $\mathcal S_{\mathrm M}\cup\mathcal S_{\mathrm A}\subseteq\mathcal S_{\mathrm{adaptive}}$.
\end{remark}

\begin{proposition}[Adaptive consent is strictly more permissive]\label{prop:adaptive}
There is an instance solvable under adaptive consent but under neither pure regime: $\mathcal S_{\mathrm M}\cup\mathcal S_{\mathrm A}\subsetneq\mathcal S_{\mathrm{adaptive}}$.
\end{proposition}
\begin{proof}
Embed E1 and E2 in orthogonal coordinate blocks of $d=4$: agents $1,2$ and directions $u_1,u_2$ act on coordinates $1,2$ exactly as in E1 ($R(u_1)=\{1\}$, $R(u_2)=\{2\}$, $\tau_1=\tau_2=0.8$), and agents $3,4$ with directions $u_3=e_3$, $u_4=e_4$ act on coordinates $3,4$ exactly as in E2 ($R(u_3)=\{3\}$, $R(u_4)=\{3,4\}$, $\tau_3=0.8$, $\tau_4=0.05$). Requirements stack axis-wise ($x_1=1$, $x_2=0$, $x_3=1$, $x_4=0.3$), each cylindrical over its owner's axis, so the cover holds. Conception reduces block-wise --- it is computed inside $P_i$, and every direction lies in a single block --- so the closure contains the plan running the E2 block first and the E1 relay second,
\[
p^{\star}=\big((u_3,1),\,(u_4,0.3),\,(u_1,1/\cos\alpha),\,(u_2,\tan\alpha)\big),
\]
extended by agents $3,4,1,2$ in turn, with endpoint $(1,0,1,0.3)\in O$.

Ratify $p^{\star}$ under $\sigma=(\mathrm A,\mathrm A,\mathrm M,\mathrm M)$. At steps 1--2 the cumulative displacement still lies entirely in block 2, so the audits take exactly E2's values: $I_3(D_1)=1$, $I_3(D_2)=1/1.09\approx0.917\ge0.8$, and $I_4(D_2)=0.09/1.09\approx0.083\ge0.05$. Steps 3--4 are memoryless, hence state-independent: $I_1(u_1)=\cos^2 20^{\circ}\approx0.883\ge0.8$ and $I_2(u_2)=1$. All consents pass and $\STOP$ holds at the endpoint, so the instance is adaptive-solvable. Note that interpretability does \emph{not} reduce block-wise --- cumulative mass dark to the auditor dilutes $I_j(D)$ (Remark~\ref{rem:mechanism}) --- so the ordering is forced: with the E1 block first, the $u_4$ step fails under both modes, $I_3(u_4)=0<0.8$ memorylessly and, at its audit, $D=(1,0,1,0.3)$ gives $I_3(D)=1/2.09\approx0.478<0.8$ and $I_4(D)=0.09/2.09\approx0.043<0.05$.

Neither pure regime solves the instance. \emph{Memoryless}: any plan reaching $O$ must move coordinate 4, hence contains a $u_4$ step, and $3\in R(u_4)$ with $I_3(u_4)=0<\tau_3$. \emph{Audited}: any plan reaching $O$ has total $u_1$-duration $1/\cos\alpha$ and a last $u_2$ step, and the E1 argument of Theorem~\ref{thm:separation} applies to agent 2's audit there, with two adaptations: $u_2$ is conceivable only at strictly positive $e_2$-coordinate, so positive $u_1$-duration precedes the last $u_2$ step and $1-\mu\cos\alpha>0$; and block-2 mass only enlarges $\|P_2^{\perp}D\|$, strengthening the audit's demand $|v|\ge2(1-\mu\cos\alpha)$ against the generation bound $|v|<(1-\mu\cos\alpha)\tan\alpha$. The contradiction $\tan\alpha>2$ persists at $\alpha=20^{\circ}$.
\end{proof}

\begin{proposition}[Cone shrinkage; forced non-plenary steps]\label{prop:cones}
$\bigcap_{j\in R}K_j$ is antitone in $R$. The two consent models confine different objects at plenary steps ($R(u_t)=N$):
\emph{(i) Memoryless.} Every executed step with coalition $R$ has \emph{direction} in $\bigcap_{j\in R}K_j$; hence an all-plenary memoryless plan has $D_T\in\cone\big(U\cap\bigcap_{j\in N}K_j\big)$, and
$$\Delta\cap\cone\Big(U\cap\bigcap_{j\in N}K_j\Big)=\varnothing\ \Longrightarrow\ \text{every memoryless plan reaching }O\text{ contains a non-plenary step.}$$
\emph{(ii) Audited.} Every plenary step constrains the \emph{cumulative displacement}: $D_t\in\bigcap_{j\in N}K_j$ whenever $R(u_t)=N$; hence an all-plenary audited plan has $D_t\in\bigcap_{j\in N}K_j$ for all $t$, in particular $D_T$, and
$$\Delta\cap\bigcap_{j\in N}K_j=\varnothing\ \Longrightarrow\ \text{every audited plan reaching }O\text{ contains a non-plenary step.}$$
\end{proposition}
\begin{proof}
Antitonicity is set intersection over larger index sets. (i): memoryless executability of a plenary step is $u_t\in\bigcap_{j\in N}K_j$; additivity gives $D_T\in\cone(U\cap\bigcap_jK_j)$ for all-plenary plans; take the contrapositive. (ii): audited executability of a plenary step is $I_j(D_t)\ge\tau_j$ for every $j$, i.e.\ $D_t\in\bigcap_{j\in N}K_j$; in particular $D_T\in\bigcap_{j\in N}K_j$ for all-plenary plans; take the contrapositive. The clauses instantiate Theorem~\ref{thm:separation}'s mechanism: memoryless gates action directions, audited gates cumulative trajectory states.
\end{proof}

Two further facts about the model --- that E1's relay is fully consented by its executor yet purposeless to him (loop inversion), and that removing state-dependence from conception collapses memoryless CCS to a single conic-feasibility test --- are recorded as Proposition~\ref{prop:loop} and Proposition~\ref{prop:special} in Appendix~\ref{app:extra}, since neither is needed for the results that follow.

\begin{remark}[Public arithmetic; stall and sound stop]\label{rem:mech}
Tokens, bases, and requirements are public, so every agent computes every $I_j$-value exactly; the only queried objects are consent values ($\tau_j$ private), truthfully returned. The qualification structure is likewise fixed by the public data: $Q_k$ is determined by the task and the representational structure, and is established by the planning procedure (Step~0 of Section~\ref{sec:algorithm}) rather than by any individual's meta-cognition --- an agent need only perform the qualified checks for the requirements on which that agent is itself qualified, not compute the full qualification structure of the team. Both ends of a plan are mechanically gated: a step whose required implementer fails the consent condition does not execute (the plan stalls there --- sovereign execution enforces itself), and termination runs through the qualified checks, sound and complete for their requirements, so a declared completion is a completion.
\end{remark}

\section{Exhaustive solvability scheme, correctness, and complexity}\label{sec:algorithm}

For the fixed horizon $h$, recall that $L_h$ is the relay closure of Definition~\ref{def:closure} truncated at depth $h$. The mathematical scheme below is \emph{exhaustive}: Step~1 ranges over every instantiated plan in $L_h$, including every duration choice satisfying the projected-progress clause. This extensional formulation separates correctness from the separate question of whether $L_h$ admits an efficient finite representation.

\begin{figure}[H]
\centering
\begin{tikzpicture}[
  node distance=4mm and 8mm,
  font=\footnotesize,
  proc/.style={draw, rounded corners, align=center, text width=8.6cm, minimum height=0.85cm, fill=blue!6, inner sep=1.5mm},
  term/.style={draw, rounded corners, align=center, text width=1.7cm, minimum height=0.75cm, fill=red!8, inner sep=1mm},
  win/.style={draw, rounded corners, align=center, text width=2.5cm, minimum height=0.75cm, fill=green!10, inner sep=1mm},
  arr/.style={-{Latex[length=1.9mm]}, thick}
]
\node[proc] (s0) {\textbf{Step 0 --- coverage preflight:} determine $Q_k$ for every $O_k$};
\node[term, right=14mm of s0] (b0) {return $\bot$\\ \scriptsize (Thm.~\ref{thm:dark})};

\node[proc, below=9mm of s0] (s1) {\textbf{Step 1 --- exhaustive relay closure:} build $L_h$ (Def.~\ref{def:closure}); keep endpoints in $O$};
\node[term, right=14mm of s1] (b1) {return $\bot$};

\node[proc, below=9mm of s1] (s2) {\textbf{Step 2 --- ratification:} query consent per step (memoryless / audited / adaptive, Rem.~\ref{rem:adaptive}); drop refused candidates};

\node[proc, below=9mm of s2] (s3) {\textbf{Step 3 --- terminal inspection:} each surviving plan, a qualified $i\in Q_k$ checks $O_k$, for every $k$};
\node[win, right=14mm of s3] (w3) {return plan $p$\\ \scriptsize sound \& valid witness};
\node[term, below=9mm of s3] (b3) {return $\bot$};

\draw[arr] (s0) -- node[above,font=\scriptsize]{empty} (b0);
\draw[arr] (s0) -- node[left,font=\scriptsize]{covered} (s1);
\draw[arr] (s1) -- node[above,font=\scriptsize]{none} (b1);
\draw[arr] (s1) -- node[left,font=\scriptsize]{found} (s2);
\draw[arr] (s2) -- (s3);
\draw[arr] (s3) -- node[above,font=\scriptsize]{passes} (w3);
\draw[arr] (s3) -- node[left,font=\scriptsize]{none pass} (b3);
\end{tikzpicture}
\caption{The exhaustive horizon-bounded solvability scheme of this section. Branch labels abbreviate the exit condition: \emph{empty} is $Q_k=\varnothing$ for some $k$; \emph{covered} is the cover holding; \emph{none}/\emph{found} refer to candidate plans reaching $O$; \emph{passes}/\emph{none pass} refer to a surviving plan clearing every qualified check in Step~3. Steps~0--1 are generative and never reject a true witness (Theorem~\ref{thm:correctness}); Steps~2--3 filter candidates against the model's exact predicates. A restricted search that explores only $\widehat L_h\subseteq L_h$ in Step~1 remains sound on any plan it returns, but a $\bot$ from such a search is not a certificate of unsolvability (Corollary~\ref{cor:restricted}).}
\label{fig:scheme}
\end{figure}
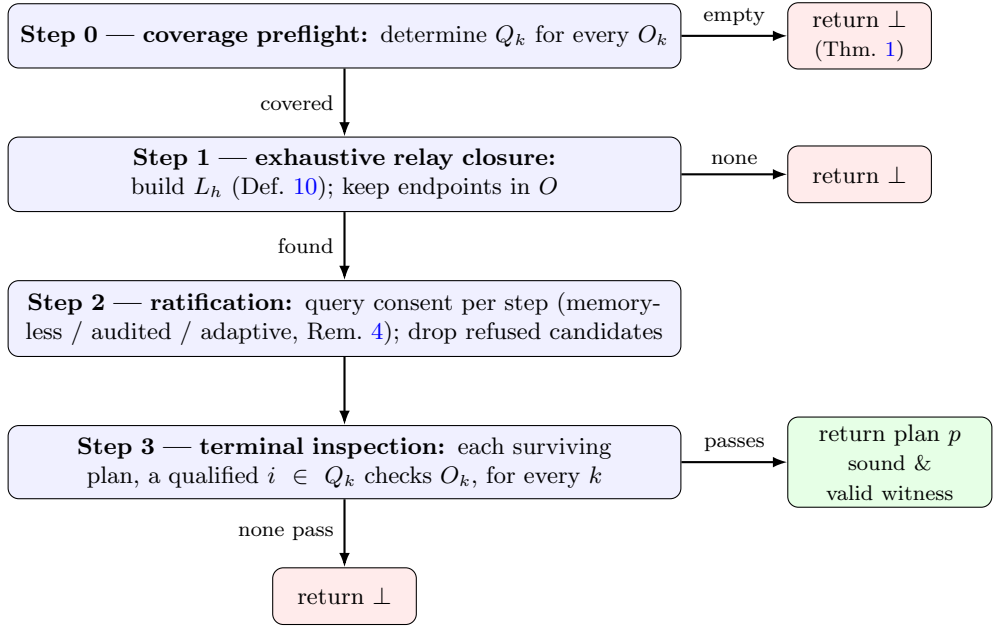

\paragraph{Step 0 --- coverage preflight (state-independent).}
Determine each $Q_k$ from the public task and representational structure. If some $Q_k=\varnothing$, return $\bot$ immediately: by Theorem~\ref{thm:dark}, the task cannot be validly completed by this team, regardless of trajectory. Confirm you have an inspector before anyone lifts a hammer.

\paragraph{Step 1 --- exhaustive relay closure (generation).}
Construct the complete horizon-bounded closure $L_h$ by Definition~\ref{def:closure}: every agent re-expands from every generated node, and every improving duration allowed by the definition is retained. Candidate solutions are the generated plans whose endpoints lie in $O$. If there is no such candidate, return $\bot$.

\paragraph{Step 2 --- ratification.}
For each candidate plan, query each required agent's consent value per step (memoryless: on $u_t$; audited: on the realized $D_t$); delete refused candidates. Audited feasibility of a candidate is the constraint system $x_T\in O$ and $\|P_j^{\perp}D_t\|\le\tan\theta_j\|P_jD_t\|$ for all $t$ and all $j\in R(u_t)$. Any symbolic or numerical manipulation of durations must preserve the conception inequalities along the realized trajectory, or the resulting plan is not a member of $L_h$. Thresholds are private, so synthesis is interactive (open problem O2).

\paragraph{Step 3 --- terminal inspection.}
For each surviving plan, every requirement $k$ is inspected by some $i\in Q_k$, who checks $P_ix_T\in P_iO_k$ in his own projection; different requirements may be covered by different agents. If all checks pass, return the plan. If none passes, return $\bot$.

\begin{theorem}[Correctness of the exhaustive horizon-bounded scheme]\label{thm:correctness}
Assume that Step~0 determines the qualification sets $Q_k$ exactly, Step~1 represents the complete closure $L_h$, and the consent and projected-membership tests in Steps~2--3 are evaluated exactly. Then the four-step scheme returns a plan if and only if CCS is yes for the given horizon $h$. Every returned plan is a valid CCS witness.
\end{theorem}
\begin{proof}
\emph{Soundness.} Suppose the scheme returns a plan $p$. Step~1 guarantees $p\in L_h$. Step~2 retains $p$ only if every required implementer consents at every step, so every step executes under Definition~\ref{def:consent}. Step~3 returns $p$ only if, for every requirement $O_k$, some qualified verifier $i\in Q_k$ finds $P_ix_T\in P_iO_k$. By Definition~\ref{def:stop}, $\STOP(x_T)$ holds. Hence $p$ satisfies Definition~\ref{def:ccs}.

\emph{Completeness.} Suppose CCS is yes. Then there exists a witness $p^{\star}\in L_h$ whose every step executes and for which $\STOP(x_T^{\star})$ holds. Because $\STOP$ requires a qualified verifier for every requirement, Step~0 does not reject. Completeness of Step~1 places $p^{\star}$ among the candidates (and $\STOP(x_T^{\star})$ implies $x_T^{\star}\in O$). Since every step of $p^{\star}$ executes, Step~2 does not delete it. Since $\STOP(x_T^{\star})$ holds, every requirement has a passing qualified check in Step~3. Thus the scheme returns a plan.
\end{proof}

\begin{corollary}[Soundness of restricted search]\label{cor:restricted}
Let a practical implementation explore any subset $\widehat L_h\subseteq L_h$ but apply ratification and terminal inspection exactly. Every plan it returns is a valid CCS witness. However, returning $\bot$ is a certificate of CCS $=$ no only when the explored set is complete, $\widehat L_h=L_h$.
\end{corollary}
\begin{proof}
The soundness argument of Theorem~\ref{thm:correctness} uses only membership in $L_h$, exact consent, and exact terminal inspection, so it applies to every returned plan from $\widehat L_h$. Completeness may fail when a valid witness lies in $L_h\setminus\widehat L_h$.
\end{proof}

\paragraph{Computational observations and open problems.}
The general model places no representation constraints on the requirements $O_k\subseteq\R^d$, and for arbitrary subsets membership, projected distance, and the qualification test need not be computable. For algorithmic and complexity statements we therefore restrict to \emph{effectively represented} requirements --- those for which these operations are computable --- with affine or polyhedral requirements as the principal example. This makes the input finite, but it does \emph{not} by itself make the continuous duration branching of $L_h$ finite; exact symbolic representation of the state-dependent relay closure remains part of the computational problem. In the state-independent comparison variant, memoryless CCS collapses to a single conic-feasibility test (Proposition~\ref{prop:special}), LP-shaped \citep{boyd2004convex}. In general, prefix-dependent enabling mirrors a central source of combinatorial search in classical planning \citep{bylander1994computational}, while audited CCS additionally imposes nonconvex cumulative-trajectory constraints. \textbf{O1}: the exact representation and complexity of state-dependent relay closure --- including which finite length or symbolic bounds survive. \textbf{O2}: the query complexity of interactive plan synthesis under private thresholds.

\section{Computational illustration}\label{sec:illustration}

A reference implementation accompanies the paper for affine requirements. To keep the search finite, it explores a restricted subset $\widehat L_h\subseteq L_h$ by choosing the projected-optimal improving $\lambda$ at each extension. Every generated token still satisfies Definition~\ref{def:closure}, and ratification and terminal inspection use the model's exact predicates; therefore every returned plan is sound by Corollary~\ref{cor:restricted}. The implementation is \emph{not} claimed complete for general CCP: a valid solution may require a different improving duration. Every numerical claim of Examples E1 and E2 is an executable assertion in the test suite. Figure~\ref{fig:e1} draws the E1 relay geometry; Figure~\ref{fig:tiers} sweeps structured random instances across span coverage and consent thresholds; Figure~\ref{fig:subteam} varies the shared representational core. All three are structured illustrations of the formal results, not empirical validation.

\begin{figure}[t]\centering
\includegraphics[width=.62\linewidth]{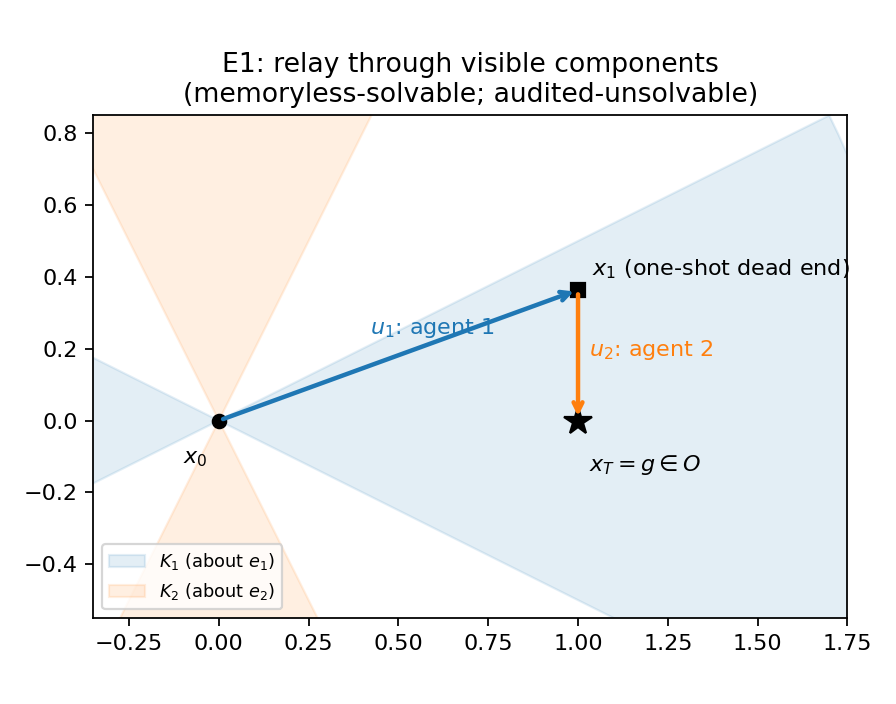}
\caption{A two-step counterfactual relay in E1. Agent 1 changes the state so that agent 2 can conceive a continuation that was unavailable at the start. In agent 2's projection the whole plan is the closed loop $0\to\tan\alpha\to0$, even though the group makes goal-relevant progress along a dimension dark to agent 2. The same instance is memoryless-solvable but audited-unsolvable (Theorem~\ref{thm:relay}, Proposition~\ref{prop:loop}, Theorem~\ref{thm:separation}).}
\label{fig:e1}
\end{figure}

\begin{figure}[t]\centering
\includegraphics[width=.95\linewidth]{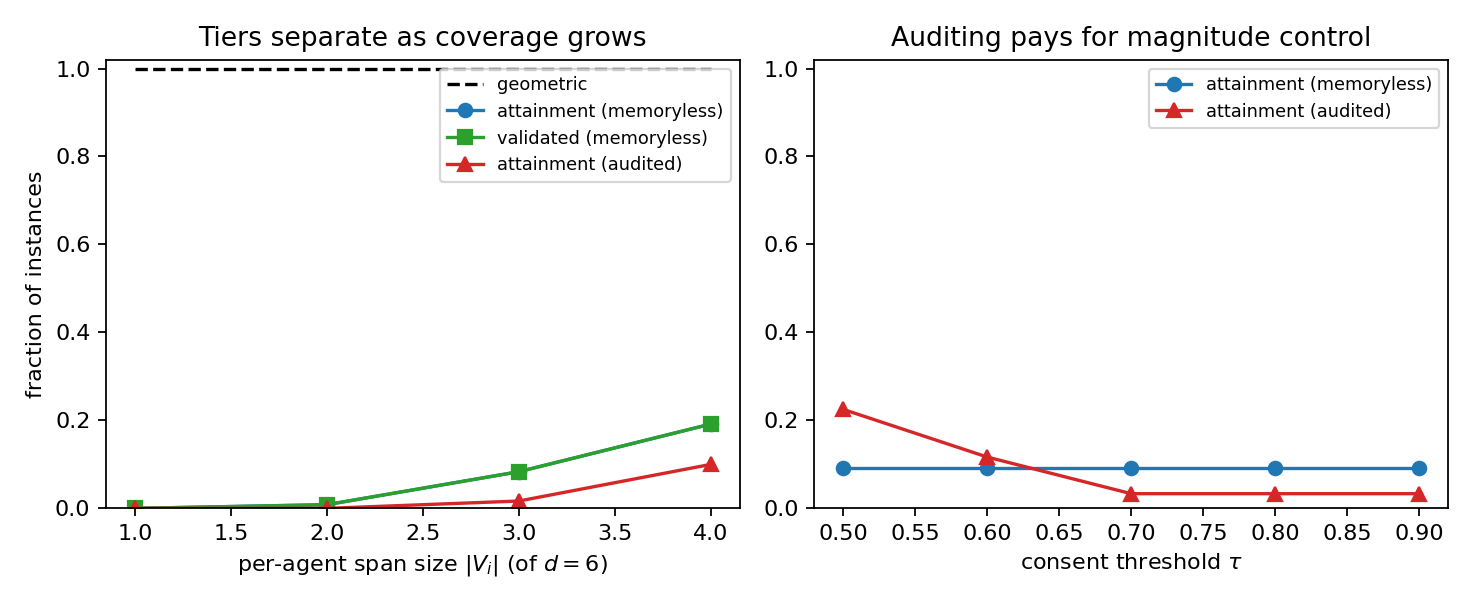}
\caption{Success rates across structured random instances. Left: greater representational coverage increases attainment and validated completion. Right: audited consent helps at lenient thresholds but hurts at strict thresholds, illustrating its incomparability with memoryless consent.}
\label{fig:tiers}
\end{figure}

\begin{figure}[t]\centering
\includegraphics[width=.6\linewidth,trim=0 0 0 40bp,clip]{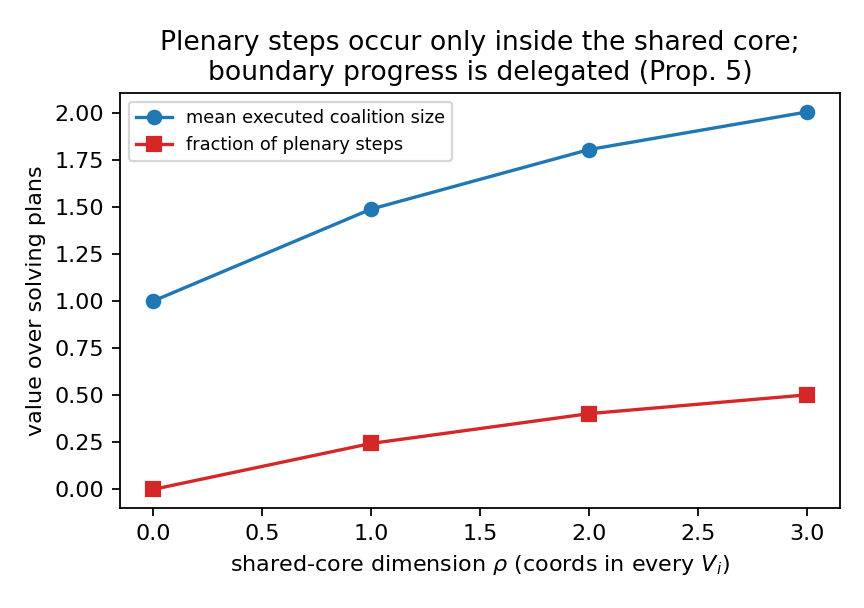}
\caption{Shared representation changes who can act together. As the common representational core grows, successful plans use larger coalitions and more whole-team steps; progress outside that core is carried by smaller subteams. All instances remain solvable.}
\label{fig:subteam}
\end{figure}

\section{Discussion}\label{sec:discussion}

\paragraph{What the results formalize.}
The first phenomenon is \emph{sequential mutual enabling}: a team can fail under one-shot pooling yet succeed when members repeatedly reconsider the evolving state, because one person's move can make another person's next move newly conceivable (Theorem~\ref{thm:relay} and Corollary~\ref{cor:visible}). The second is \emph{competence without global purpose}: an executor may fully understand the step he performs while the reason the step matters lies outside his projection (Proposition~\ref{prop:loop}). The third is \emph{forced sub-teaming}: when plenary participation is confined by the shared representational core, boundary progress may require delegation to proper subcoalitions (Proposition~\ref{prop:cones}). The negative limit is equally sharp: if a task requirement varies along a dimension dark to the whole team, the team cannot validly declare completion even if it accidentally reaches the right physical state (Theorem~\ref{thm:dark}). These phenomena are documented, separately, in distributed cognition \citep{hutchins1995cognition}, epistemic dependence \citep{hardwig1985epistemic}, and high-reliability organizing \citep{WeickSutcliffeObstfeld1999CollectiveMindfulness,vaughan1996challenger}; CCP supplies explicit geometric conditions under which they arise.

\paragraph{Scope.}
Three deliberate exclusions define the regime. No strategy: agents are truthful and share $O$; this is the cognitive skeleton of cooperation with the political flesh removed. No learning: bases are fixed, so the model covers collaboration on timescales shorter than teaching; changes to the representational basis lie outside the present model. And sovereign execution picks out professionalized collaboration --- ``I do not perform steps I do not understand'' is the working ethics of medicine, aviation, and licensed trades --- rather than authority-gated organizations. Within the regime, the claims are face-valid, not validated: whether real teams fail \emph{by these mechanisms} is an empirical bet, and the model's prediction shapes (sub-teaming density at basis boundaries; verification staffing predicting valid completion) say where to test it.

\paragraph{Future work.}
The hardness of general CCS (O1, O2); and team synthesis --- given a task $\{O_k\}$ and a pool of candidate agents with subspaces and thresholds, construct the minimum team for which CCS is yes --- the representational successor of team formation with required skills \citep{lappas2009finding,juarez2021comprehensive}.

\section{Conclusion}\label{sec:conclusion}

Collective Counterfactual Planning models a team as a set of projections of one task space, and derives relay, consent, sub-teaming, and sign-off from representational geometry interacting with the task's implementation structure. Its positive result explains how a group can reach a goal no member can plan alone; its negative result identifies requirements that no amount of effort, honesty, or luck can let the team validly finish. The exhaustive horizon-bounded solvability scheme is correct for the formal CCS problem, while practical restricted searches remain sound but may be incomplete. You do not need a team that can build every possible house. You need people whose combined representational and practical competencies suffice to plan, execute, and verify this one.

\paragraph{Reproducibility.} Code and reproduction materials are available at \url{https://github.com/DarkEyes/Coll-Counterfactual-Plan}. All numerical claims of Examples E1 and E2, and of Proposition~\ref{prop:adaptive}, are executable assertions in the accompanying test suite (\texttt{test\_e1.py}); the figures are produced by \texttt{experiments.py} from the library \texttt{ccp.py}, seeded for reproducibility. The code implements the restricted duration policy described in Section~\ref{sec:illustration}; it is used to check witnesses and illustrate the theory, not as a completeness certificate for general CCS. The audited$\setminus$memoryless direction of Theorem~\ref{thm:separation} was discovered by this code: the first threshold sweep contradicted a draft containment claim, and E2 is the distilled witness. Proposition~\ref{prop:adaptive}'s witness was likewise code-checked: an initial draft ordering failed to ratify under every consent policy, since audited interpretability does not decompose across orthogonal blocks (Remark~\ref{rem:mechanism}); the corrected ordering in the proof is the one the test suite confirms.

\subsection*{Use of generative AI and AI-assisted technologies}
During the preparation of this work the author used Claude (Anthropic)/chatGPT (OpenAI) to edit and polish the draft. After using these tools, the author reviewed and edited the content as needed and takes full responsibility for the content of the publication.

\appendix
\section{Additional results}\label{app:extra}

These two observations sharpen the reading of E1 and of the closure mechanism, respectively, but neither is needed for the paper's main line of argument (Sections~\ref{sec:results}--\ref{sec:algorithm}), so they are collected here rather than in the main text.

\begin{proposition}[Loop inversion]\label{prop:loop}
In the E1 relay plan, executor 2 fully interprets his step ($I_2(u_2)=1$), yet the plan's point is dark to him: his projected view of the whole trajectory is the closed loop $0\to\tan\alpha\to0$, $P_2(x_T-x_0)=0$, and the goal-relevant progress lies along $e_1$ with $P_2e_1=0$. Interpretability gates motion, never purpose.
\end{proposition}
\begin{proof}
By the E1 computations.
\end{proof}

\begin{proposition}[State-independent comparison variant]\label{prop:special}
Under the concrete progress rule \eqref{eq:C}, availability is generically state-dependent; for comparison, consider the \emph{variant} in which availability is exogenously state-independent, $C_i(x,O)\equiv C_i$ (replacing Definition~\ref{def:conception}). In this variant, under memoryless consent, if $\Delta\cap\cone(U^{\ast})\neq\varnothing$ and the verification cover holds, where
$U^{\ast}=\{u:\ \forall j\in R(u),\,u\in K_j\ \text{ and }\ \exists i\in R(u),\,u\in C_i\}$,
then a witnessing plan with $T\le d$ exists; hence for every horizon $h\ge d$, CCS holds iff the cover holds and $\Delta\cap\cone(U^{\ast})\neq\varnothing$.
\end{proposition}
\begin{proof}
Availability and consent are then prefix-independent, so appendability and executability of a token depend only on its direction, and reachable endpoints are exactly $x_0+\cone(U^{\ast})$ by additivity. By conic Carath\'eodory \citep{rockafellar1970convex}, any point of $\cone(U^{\ast})$ is a nonnegative combination of at most $d$ linearly independent members; merging equal-direction tokens gives $T\le d$, which fits within any horizon $h\ge d$. The variant isolates exactly what prefix-dependent enabling contributes: with it removed, order-irrelevance returns, reachable endpoints form a cone, and short plans suffice.
\end{proof}

\bibliographystyle{plainnat}
\bibliography{ccp}

\end{document}